\documentclass[11pt]{elsarticle}
\usepackage{setspace}
\usepackage[plain]{fullpage}
\usepackage[utf8]{inputenc}
\usepackage{amsthm, amsmath, amssymb, amsfonts}
\usepackage{color}
\usepackage{algorithm}
\usepackage{algorithmicx}
\usepackage[noend]{algpseudocode}
\usepackage{algpascal}

\usepackage{subcaption}
\usepackage{graphics,graphicx,caption}

\usepackage{amsmath}
\usepackage{amssymb}
\usepackage{color}
\definecolor{red}{RGB}{255,0,0}
\definecolor{blue}{RGB}{0,0,255}
\definecolor{green}{RGB}{0,255,0}

\newcommand {\abs}[1]  {\left\vert#1\right\vert}
\newcommand {\set}[1]  {\left\{#1\right\}}

\newcommand{\itemref}[1]{{\textrm{\ref{itm:#1}})}}

\newtheorem{theorem} {Theorem}
\newtheorem{lemma}   [theorem]  {Lemma}
\newtheorem{definition}         {Definition}
\newtheorem{corollary}[theorem] {Corollary}
\newtheorem{observation}[theorem] {Observation}
\newtheorem{proposition}[theorem] {Proposition}

\newcommand{\cg}{{\cal G}}

\newcommand{\cs}{{\cal S}}
\newcommand{\bigoh}[1] {{$\mathcal{O}$}$\left(#1\right)$}

\begin{document}
\begin{frontmatter}
\title{Equimatchable Claw-Free Graphs\tnoteref{TR-SLO}}
\tnotetext[TR-SLO]{The support of 213M620 Turkish-Slovenian TUBITAK-ARSS Joint Research Project is greatly acknowledged.}

\author[sharif]{Saieed Akbari\fnref{saieed}}
\ead{s\_akbari@sharif.edu}
\fntext[saieed]{Part of this research was carried out while Saieed Akbari was visiting Istanbul Center for Mathematical Sciences (IMBM) whose support is greatly acknowledged.}

\author[gtu]{Hadi Alizadeh}
\ead{halizadeh@gtu.edu.tr}

\author[bu]{T{\i}naz Ekim}
\ead{tinaz.ekim@boun.edu.tr}

\author[gtu]{Didem G{\"o}z{\"u}pek}
\ead{didem.gozupek@gtu.edu.tr}

\author[bu,telhai]{Mordechai Shalom\fnref{mordo}}
\ead{cmshalom@telhai.ac.il}
\fntext[mordo]{The work of this author is supported in part by the TUBITAK 2221 Programme.}

\address[sharif]{Department of Mathematical Sciences, Sharif University of Technology, 11155-9415, Tehran, Iran}
\address[gtu]{Department of Computer Engineering, Gebze Techical University, Kocaeli, Turkey}
\address[bu]{Department of Industrial Engineering, Bogazici University, Istanbul, Turkey}
\address[telhai]{TelHai Academic College, Upper Galilee, 12210, Israel}

\cortext[corr]{Corresponding Author}

\begin{abstract}
A graph is equimatchable if all of its maximal matchings have the same size. A graph is claw-free if it does not have a claw as an induced subgraph. In this paper, we provide the first characterization of claw-free equimatchable graphs by identifying the equimatchable claw-free graph families. This characterization implies an efficient recognition algorithm. 

\end{abstract}

\end{frontmatter}

\section{Introduction}
A graph $G$ is \emph{equimatchable} if every maximal matching of $G$ has the same cardinality. Equimatchable graphs are first considered by Grünbaum \cite{grunbaum}, Lewin \cite{lewin},  and Meng \cite{meng} simultaneously in 1974. They are formally
introduced by Lesk et al. in 1984 \cite{equi-LPP83}. Equimatchable graphs can be recognized in polynomial time (see \cite{equi-LPP83} and \cite{demange_ekim_equi}). From the structural point of view, all 3-connected planar equimatchable graphs and all  3-connected cubic equimatchable graphs are determined by Kawarabayashi et al. \cite{classes}. Besides, Kawarabayashi and Plummer showed that equimatchable graphs with fixed genus have bounded size \cite{genus}, while Frendrup et al. characterized equimatchable
graphs with girth at least 5 \cite{girth}. Factor-critical equimatchable graphs with vertex connectivity 1 and 2 are characterized by Favaron \cite{favaron1986equimatchable}.

A graph $G$ is \emph{well-covered} if every maximal independent set of $G$ has the same cardinality. Well-covered graphs are closely related to equimatchable graphs since the line graph of an equimatchable graph is a well-covered graph. Finbow et al. \cite{finbow1994characterization} provides a characterization of well-covered graphs that contain neither 4-cycles nor 5-cycles, whereas Staples \cite{staples1979some} provides characterizations of some subclasses of well-covered graphs. A graph is \emph{claw-free} if it does not have a claw as an induced subgraph. Recognition algorithms for claw-free graphs have been presented by Kloks et al. \cite{Kloks2000}, Faenza et al. \cite{faenza2014solving}, and Hermelin et al. \cite{hermelin2011domination}.
Claw-free well-covered graphs have been investigated by Levit and Tankus \cite{levit2015weighted} and by Hartnell and Plummer \cite{hartnell1996}. However, to the best of our knowledge, there is no previous study in the literature about claw-free equimatchable graphs.

In this paper, we investigate the characterization of claw-free equimatchable graphs. In Section \ref{sec:Preliminaries}, we give some preliminary results. In particular, we show that the case of equimatchable claw-free graphs with even number of vertices reduces to cliques with an even number of vertices or a 4-cycle, and all graphs with odd number of vertices and independence number at most 2 are claw-free and equimatchable. We also show that the remaining equimatchable claw-free graphs have (vertex) connectivity at most 3. Based on this fact, in Section \ref{sec:OddOrder}, we focus on 1-connected, 2-connected (based on a result of Favaron \cite{favaron1986equimatchable}) and 3-connected equimatchable claw-free graphs with odd number of vertices separately. Our full characterization is summarized in Section \ref{sec:Summary}, where we provide a recognition algorithm running in time \bigoh{m^{1.407}} where $m$ refers to the number of edges in the input graph.

\section{Preliminaries}\label{sec:Preliminaries}
In this section, after giving some graph theoretical definitions, we mention some known results about matchings in claw-free graphs and develop some tools for our proofs.

Given a simple graph $G=(V(G), E(G))$, a \emph{clique} (resp. \emph{independent set}) of $G$ is a subset of pairwise adjacent (resp. non-adjacent) vertices of $G$. The \emph{independence number} of $G$ denoted by $\alpha(G)$ is the maximum size of an independent set of $G$. We denote by $N(v)$ the set of neighbors of $v \in V(G)$.
For a subgraph $G'$ of $G$, $N_{G'}(v)$ denotes $N(v) \cap V(G')$. A vertex $v$ is \emph{complete to} a subgraph $G'$ if $N_{G'}(v)=V(G')$. For  $U \subseteq V(G)$, we denote by $G[U]$ the subgraph of $G$ induced by $U$. For simplicity, according to the context, we will use a set of vertices or the (sub)graph induced by a set of vertices in the same manner.  We denote by $uv$ a potential edge between two vertices $u$ and $v$. Similarly, we denote paths and cycles of a graph as sequences of its vertices. In this work, $n$ denotes the \emph{order} $\abs{V(G)}$ of the graph $G$. We say that $G$ is an \emph{odd graph} (resp. \emph{even graph}) if $n$ is odd (resp. even). For a set $X$ and a singleton $\set{x}$ we use the shorthand notations $X + x$ and $X - x$ for $X \cup \set{x}$ and $X \setminus \set{x}$, respectively.

We denote by $P_p$, $C_p$ and $K_p$ the path, cycle, and complete graph, respectively, on $p$ vertices and by $K_{p,q}$ the complete bipartite graph with bipartition sizes $p$ and $q$.
The graph $K_{1,3}$ is termed \emph{claw}. A graph is \emph{claw-free} if it contains no claw as an induced subgraph.

A set of vertices $S$ of a connected graph $G$ such that $G \setminus S$ is not connected is termed a \emph{cut set}. A cut set is  \emph{minimal} if none of its proper subsets is a cut set. A \emph{$k$-cut} is a cut set with $k$ vertices. A graph is \emph{$k$-connected} if it has more than $k$ vertices and every cut set of it has at least $k$ vertices. The \emph{(vertex) connectivity} of $G$, denoted by $\kappa(G)$, is the biggest number $k$ such that $G$ is $k$-connected.

A \emph{matching} of a graph $G$ is a subset $M \subseteq E(G)$ of pairwise non-adjacent edges. A vertex $v$ of $G$ is \emph{saturated} by $M$ if $v \in V (M)$ and \emph{exposed} by
$M$ otherwise. A matching $M$ is \emph{maximal} in $G$ if no other matching of $G$ contains $M$. Note that a matching $M$ is maximal if and only if $V(G) \setminus V(M)$ is an independent set. A matching $M$ is a \emph{perfect matching} of $G$ if $V(M)=V(G)$.

A graph $G$ is \emph{equimatchable} if every maximal matching of $G$  has the same cardinality. A graph $G$ is \emph{randomly matchable} if every matching of $G$  can be extended to a perfect matching. In other words, randomly matchable graphs are equimatchable graphs admitting a perfect matching. A graph $G$  is \emph{factor-critical} if $G - u$ has a perfect matching for every vertex $u$ of $G$.

The following facts are frequently used in our arguments.

\begin{lemma}\label{lem:ClawFreeAdmitsPerfectMatching}
\cite{sumner1974graphs} Every connected claw-free even graph admits a perfect matching.
\end{lemma}

\begin{corollary} \label{coro:ClawFreeTwoConnectedIsFactorCritical}
\cite{sumner1974graphs} Every 2-connected claw-free odd graph is factor-critical.
\end{corollary}

\begin{lemma}\label{lem:RandomlyMatchable}
\cite{sumner1979randomly} A connected graph is randomly matchable if and only if it is isomorphic to $K_{2p}$ or $K_{p,p}$ ($p \ge 1$).
\end{lemma}

Using the above facts, we identify some easy cases as follows.

\begin{proposition}\label{prop:EvenEquimatchable}
A connected even graph is claw-free and equimatchable if and only if it is  isomorphic to $K_{2p}$ ($p \geq 1$) or $C_4$.
\end{proposition}
\begin{proof}
The graphs $K_{2p}$ and $C_4$ are clearly equimatchable and claw-free. Conversely, let $G$ be a connected equimatchable claw-free even graph. By Lemma \ref{lem:ClawFreeAdmitsPerfectMatching}, $G$ admits a perfect matching. Therefore, $G$ is randomly matchable. By Lemma \ref{lem:RandomlyMatchable}, $G$ is either a $K_{p,p}$ or a $K_{2p}$ for some $p \geq 1$. Since $G$ is a claw-free graph, it is a $K_{2p}$ or a $C_4$.
\end{proof}

\begin{lemma}\label{lem:EibenAlphaTwo}
Every odd graph $G$ with $\alpha(G)=2$ is equimatchable and claw-free.
\end{lemma}
\begin{proof}
Every matching of $G$ has at most $(n-1)/2$ edges since $n$ is odd. On the other hand, a maximal matching with less than $(n-1)/2$ edges implies an independent set with at least 3 vertices, a contradiction. Then every maximal matching has exactly $(n-1)/2$ edges. The graph $G$ is clearly claw-free because a claw contains an independent set with 3 vertices.
\end{proof}

Thus, from here onwards, we focus on the case where $G$ is odd and $\alpha(G) \geq 3$. The following lemmas provide the main tools to obtain our characterization in Section \ref{sec:OddOrder} and enable us to confine the rest of this study to the cases with connectivity at most 3.

\begin{lemma}\label{lem:isolating}
Let $G$ be a connected equimatchable claw-free odd graph and $M$ be a matching of $G$. Then the following hold:
\begin{enumerate}[i)]
\item \label{itm:EveryMMLeavesOneExposed} Every maximal matching of $G$ leaves exactly one vertex exposed.

\item \label{itm:OneOddEqimatchableComponent} The subgraph $G\setminus V(M)$ contains exactly one odd connected component and this component is equimatchable.

\item \label{itm:EvenComponentsRandomlyMatchable} The even connected components of $G\setminus V(M)$ are randomly matchable.
\end{enumerate}
\end{lemma}

\begin{proof}
\begin{enumerate}[i)]
\item Let $v$ be a non-cut vertex of $G$ (every graph has such a vertex). Then $G - v$ is a connected claw-free even graph, which by Lemma \ref{lem:ClawFreeAdmitsPerfectMatching} admits a perfect matching with size $(n-1)/2$. This matching is clearly a maximum matching of $G$ that leaves exactly one vertex exposed. Since $G$ is equimatchable, every maximal matching of $G$ leaves exactly one vertex exposed.

\item Since $G$ is odd and $V(M)$ has an even number of vertices, $G \setminus V(M)$ contains at least one
odd component. If $G \setminus V (M)$ contains two odd components, then every maximal matching
extending $M$ leaves at least two exposed vertices, contradicting \itemref{EveryMMLeavesOneExposed}. Let $G_1$ be the unique odd component of $G \setminus V(M)$. Assume for a contradiction that some maximal matching $M_1$ of $G_1$ leaves at least three exposed vertices. Then any
maximal matching of $G$ extending $M \cup M_1$ leaves at least three exposed vertices, contradicting
\itemref{EveryMMLeavesOneExposed}. Therefore, every maximal matching of $G_1$ leaves exactly one vertex exposed; i.e., $G_1$ is equimatchable.

\item Let $G_i$ be an even component of $G \setminus V(M)$. Assume for a contradiction that there is a maximal matching $M_i$ of $G_i$ leaving at least two exposed vertices. Then any maximal matching of $G$ extending $M \cup M_i$ leaves at least two exposed vertices, contradicting \itemref{EveryMMLeavesOneExposed}.
\end{enumerate}
\end{proof}

\begin{lemma}\label{lem:notequiIS}
Let $G$ be a connected claw-free odd graph. Then $G$ is equimatchable if and only if for every independent set $I$ of 3 vertices, $G \setminus I$ has at least two odd connected components.
\end{lemma}
\begin{proof}
As in the proof of Lemma \ref{lem:isolating} \itemref{EveryMMLeavesOneExposed}, picking up a non-cut vertex $v$ of $G$, the perfect matching of $G-v$ is a matching of $G$ with $(n-1)/2$ edges.

($\Rightarrow$) Assume that $G$ is equimatchable, and let $I$ be an independent set of $G$ with $3$ vertices. Suppose, for a contradiction, that all connected components of $G \setminus I$ are even. Thus, every such connected component admits a perfect matching by Lemma \ref{lem:ClawFreeAdmitsPerfectMatching}. The union of all these matchings is a maximal matching of $G$ with size $(n-3)/2$, contradicting the equimatchability of $G$. Then $G \setminus I$ has at least one odd component. The claim follows from parity considerations.

($\Leftarrow$) Assume that $G$ is not equimatchable. Then $G$ has a maximal matching $M$ of size $(n-3)/2$ by the following fact. Consider any maximal matching $M'$ of $G$ with size $(n-\ell)/2$ for some $\ell \geq 3$. If $\ell \ge 4$ find an $M'$-augmenting path and increase $M'$ along this augmenting path. Indeed, the new matching $M''$ obtained in this way is still maximal (the set of vertices exposed by $M''$ is a subset of vertices exposed by $M'$) and contains one more edge. We repeat this procedure until the matching reaches size $(n-3)/2$. Then $I= G \setminus V(M)$ is an independent set with size 3 and $G \setminus I$ has a perfect matching, namely $M$. This implies that every connected component of $G \setminus I$ is even.
\end{proof}

\begin{corollary}\label{cor:4conn}
If $G$ is an equimatchable claw-free odd graph with $\alpha(G) \geq 3$, then $\kappa(G) \leq 3$.
\end{corollary}
\begin{proof}
Let $I$ be an independent set of $G$ with three vertices, and assume for a contradiction that $\kappa(G)\geq 4$. Then $G \setminus I$ is connected and even, contradicting Lemma \ref{lem:notequiIS}.
\end{proof}

\section{Equimatchable Claw-Free Odd Graphs with $\alpha(G)\ge 3$}\label{sec:OddOrder}
Let $G$ be a connected equimatchable claw-free odd graph with $\alpha(G) \geq 3$. By Corollary \ref{cor:4conn}, $\kappa(G)\leq 3$. Since $\alpha(G) \geq 3$, $G$ contains independent sets $I$ of three vertices, each of which is a 3-cut by Lemma \ref{lem:notequiIS}. If $\kappa(G)=3$, then every such $I$ is a minimal cut set. In Section \ref{sec:3conn} (see Lemma \ref{lem:kappa3}) we show that the other direction also holds; i.e.\ if every such $I$ is a minimal cut set, then $\kappa(G)=3$. Therefore, if $\kappa(G)=2$, at least one independent 3-cut $I$ is not minimal; i.e.\ $G$ contains two non-adjacent vertices forming a cut set (we will call this cut set a strongly independent 2-cut). We analyze this case in Section \ref{sec:2conn}. Finally, we analyze the case $\kappa(G)=1$ in Section \ref{sec:1conn}.

In each subsection we describe the related graph families. Although we will use their full descriptions in the proofs, we also introduce the following notation  for a more compact description that will be useful in the illustrations of Figure \ref{fig:graphfamilies} and in the recognition algorithm given in Corollary \ref{cor:recognition}. Let $H$ be a graph on $k$ vertices $v_1,v_2,\ldots , v_k$ and let $n_1,n_2,\ldots , n_k$ be non-negative integers denoting the \textit{multiplicities} of the corresponding vertices. Then $H(n_1,n_2,\ldots , n_k)$ denotes the graph obtained from $H$ by repeatedly replacing each vertex $v_i$ with a clique of $n_i\geq 0$ vertices, each of which having the same neighborhood as $v_i$; i.e.\ each vertex in such a clique is a \textit{twin} of $v_i$. Clearly, $H=H(1,\ldots,1)$ where all multiplicities are 1. Note that if $n_i=0$ for some $i$, this means that the vertex $v_i$ is deleted.

The following observations will be useful in our proofs.

\begin{lemma}\label{lem:ClawFreeCutSets}
Let $G$ be a connected claw-free graph, $S$ be a minimal cut set of $G$, $C$ be an induced cycle of $G \setminus S$ with at least $4$ vertices, and $K$ be a clique of $G \setminus S$. Then
\begin{enumerate}[i)]
\item{} \label{itm:VertexAdjacentToAtmostTwoComponents} $G \setminus S$ consists of exactly two connected components, and every vertex of $S$ has a neighbour in both of them.
\item{} \label{itm:NeighborhoodOfConnectingVertexIsTwoCliques} The set $N_{G_i}(s)$ is a clique for every vertex $s \in S$ and every connected component $G_i$ of $G \setminus S$.
\item{} The neighborhood of every vertex of $S$ in $C$ is either empty or consists of exactly two adjacent vertices of $C$.
\item{} \label{itm:NeighborhoodOfCutSetInAClique} If $s_1$ and $s_2$ are two non-adjacent vertices of $S$, then $N_K(s_1) \cap N_K(s_2) = \emptyset$ or $N_K(s_1) \cup N_K(s_2) = K$.
\end{enumerate}
\end{lemma}
\begin{proof}
\begin{enumerate}[i)]
\item{} By the minimality of $S$, every vertex $s \in S$ is adjacent to at least two components of $G \setminus S$. Assume for a contradiction that a vertex $s \in S$ is adjacent to three connected components of $G \setminus S$. Then, $s$ together with one arbitrary vertex adjacent to it from each component form a claw, contradiction. Therefore, every vertex $s \in S$ is adjacent to exactly two components of $G \setminus S$. Furthermore, by the minimality of $S$, every component is adjacent to every vertex of $S$. Therefore $G \setminus S$ consists of exactly two connected components.
\item{} \label{itm:NeighborhoodIsAClique} Let $s \in S$, and $G_1, G_2$ be the two connected components of $G \setminus S$. Assume that the claim is not correct. Then, without loss of generality, there are two non-adjacent vertices $w,w' \in N_{G_1}[s]$. Then $s,w,w'$ together with an arbitrary vertex of $N_{G_2}(s)$ form a claw, contradiction.
\item{} \label{itm:TwoAdjacentNeigborsInCycle} Let $s \in S$ be adjacent to a vertex $v$ of $C$. If $s$ is adjacent to none of the two neighbors of $v$ in $C$, then $v$, $s$, and the two neighbors $v$ in $C$ form a claw, contradiction. If $s$ has three neighbors in $C$, then its neighborhood in the connected component of $C$ is not a clique, contradicting \itemref{NeighborhoodIsAClique}.
\item{} Assume for a contradiction that $N_K(s_1) \cap N_K(s_2) \neq \emptyset$ and $N_K(s_1) \cup N_K(s_2) \subset K$. Let $c \in N_K(s_1) \cap N_K(s_2)$ and $a \in K \setminus N_K(s_1) \cup N_K(s_2)$. Then $\set{s_1,s_2,a,c}$ induces a claw, contradiction.
\end{enumerate}
\end{proof}

\subsection{Equimatchable Claw-Free Odd Graphs with $\alpha(G)\geq 3$ and $\kappa(G)=3$}\label{sec:3conn}
In this section we show that the class of claw-free equimatchable odd graphs with independence number at least 3 and connectivity 3 is equivalent to the following graph class.

\begin{definition}\label{defn:GraphFamilyThreeConn}
Graph $G \in \cg_3$ if it has an independent 3-cut $S=\set{s_1,s_2,s_3}$ such that
\begin{enumerate}[i)]
\item \label{itm:G31} The subgraph $G \setminus S$ consists of two connected components $A$ and $A'$, each of which is an odd clique of at least three vertices
\item \label{itm:G32} there exist two vertices $a \in A$, $a' \in A'$ such that
\begin{itemize}
\item $N(s_1) = A + a'$,
\item $N(s_2) = A' + a$, and
\item $N(s_3) = A \cup A' \setminus \set{a, a'}$.
\end{itemize}
\end{enumerate}
\end{definition}

We note that
\begin{eqnarray*}
\cg_3 = \set{G_3(1,2p,1,1,1,2q,1)|~p,q \geq 1}
\end{eqnarray*}
where $G_3$ is the graph depicted in Figure \ref{subfig:G3}.

\begin{proposition}\label{prop:3ConnectedOnlyIf}
If $G \in \cg_3$, then $G$ is a
connected equimatchable claw-free odd graph with $\alpha(G) = \kappa(G) = 3$.
\end{proposition}
\begin{proof}
The only independent sets of $G_3$ with 3 vertices are $S=\set{s_1, s_2, s_3}$ and $S'=\set{s_3,a,a'}$. Both $G \setminus S$ and $G \setminus S'$ have two odd components; hence, $G$ is equimatchable by Lemma \ref{lem:notequiIS}. All other properties are easily verifiable.
\end{proof}

The following lemma provides the general structure of the claw-free equimatchable odd graphs with $\alpha(G)\geq 3$ and $\kappa(G) \le 3$.
\begin{lemma}\label{lem:minimalCutSetWithThree}
Let $G$ be an equimatchable claw-free odd graph. If $S=\set{s_1,s_2,s_3}$ is a minimal independent cut set of $G$, then $G \setminus S$ consists of two odd cliques $A$ and $A'$, each of which has at least three vertices, and every vertex of $S$ has a neighbor in both $A$ and $A'$.
\end{lemma}
\begin{proof}
By Lemma \ref{lem:ClawFreeCutSets} \itemref{VertexAdjacentToAtmostTwoComponents}, $G \setminus S$ consists of two components $A$ and $A'$.
By Lemma \ref{lem:notequiIS}, both $A$ and $A'$ are odd.
If one of $A$ of $A'$ consists of a single vertex, then this single vertex together with $S$ forms a claw. 
Therefore, each of $A$ and $A'$ has at least three vertices. It remains to show that both $A$ and $A'$ are cliques.

Let $v, v'$ be two vertices of $G[A]$ such that the distance between $v$ and $v'$ is as large as possible. 
If $v$ and $v'$ are adjacent, then $A$ is a clique. 
Now suppose that $vv' \notin E(G)$. 
We claim that neither $v$ nor $v'$ is a cut vertex of $G[A]$. 
Suppose that $G[A \setminus v]$ has two connected components, $B$ and $B'$. Without loss of generality, let $v'$ be in $B'$. 
Then every vertex $b$ in $B$ is further form $v'$ than $v$ is, since every path between $b$ and $v'$ contains $v$, a contradiction. 
Therefore, neither $v$ nor $v'$ is a cut vertex of $G[A]$, as claimed. 
At least one of $v,v'$ is adjacent to at most one vertex of $S$ because otherwise, by counting arguments, at least one vertex of $S$ is adjacent to both $v$ and $v'$, contradicting Lemma \ref{lem:ClawFreeCutSets} \itemref{NeighborhoodIsAClique}. Assume without loss of generality that $v$ is non-adjacent to $\set{s_1,s_2}$, and consider the independent set $I= \set{s_1, s_2, v}$. If $v$ is not the unique vertex of $A$ adjacent to $s_3$, then $G \setminus I$ is connected and even, contradicting Lemma \ref{lem:notequiIS}. Otherwise, $G \setminus I$ consists of two even components, again contradicting Lemma \ref{lem:notequiIS}. Therefore, $A$ is a clique, and by symmetry, so is $A'$.
\end{proof}

We note that Lemma \ref{lem:minimalCutSetWithThree} is a variant of the following result in the literature for the case $k=3$; indeed Lemma \ref{lem:minimalCutSetWithThree} is also valid for connectivity 1 and 2. This will enable us to replace the connectivity 3 condition with the existence of a minimal independent cut set of three vertices in what follows.
\begin{lemma}\label{lem:eiben3conn}
\cite{eiben2015equimatchable} Let $G$ be a $k$-connected equimatchable factor-critical graph with at least $2k+3$ vertices and a $k$-cut $S$ such that $G\setminus S$ has two components with at least $3$ vertices, where $k \ge 3$. Then $G\setminus S$ has exactly two components and both are complete graphs.
\end{lemma}

\begin{proposition}\label{prop:3conn}
If $G$ is an equimatchable claw-free odd graph with $\alpha(G) \ge 3$ and it contains a minimal independent cut set $S=\set{s_1,s_2,s_3}$ with three vertices, then $G \in \cg_3$.
\end{proposition}
\begin{proof}
By Lemma \ref{lem:minimalCutSetWithThree}, Property \itemref{G31} of Definition \ref{defn:GraphFamilyThreeConn} holds. We proceed to show \itemref{G32}. Since $S$ is minimal, every vertex $s \in S$ is has a neigbour in each of $A$ and $A'$. Suppose that a connected component of $G \setminus S$, say $A$, has a vertex $v$ that is non-adjacent to two vertices, say $s_1, s_2$ of $S$. Then $I=\set{s_1,s_2,v}$ is an independent set with three vertices and $G \setminus I$ is either connected, or has two even components $A'+s_3$ and $A-v$ (when $N_A(s_3)=\{v\}$), contradicting Lemma \ref{lem:notequiIS}. Therefore, every vertex of $A \cup A'$ is adjacent to at least two vertices of $S$. As already observed, a vertex of $A \cup A'$ that is complete to $S$ implies a claw, contradiction. We conclude that every vertex of $A \cup A'$ is adjacent to exactly two vertices of $S$. For $i,j \in [3]$, let $N_{i,j}=N_A(s_i) \cap N_A(s_j)$ and $N'_{i,j}=N_{A'}(s_i) \cap N_{A'}(s_j)$. We have shown that $\set{N_{1,2}, N_{2,3}, N_{1,3}}$ (resp. $\set{N'_{1,2}, N'_{2,3}, N'_{1,3}}$) is a partition of $A$ (resp. $A'$).

Assume that for some pair $(i,j)$ none of $N_{i,j}, N'_{i,j}$ is empty, and let $k=6-i-j$. Consider the set $S'=\set{s_k, w_{ij}, w'_{ij}}$ where $w_{ij}$ and $w'_{ij}$ are arbitrary vertices of $N_{i,j}$ and $N'_{i,j}$, respectively. $S'$ is an independent set, and it is easy to verify that if one of $N_{i,j}$ and $N'_{i,j}$ is not a singleton, say $N_{i,j}$, then $G\setminus S'$ is connected; indeed, in this case, there exists a vertex $u_{ij}\in N_{i,j}-w_{ij}$. Moreover, either $N'_{i,k}$ or $N'_{j,k}$ is non-empty (since otherwise $s_k$ would not have a neighbor in $A'$, contradicting the minimality of $S$) implying that $G\setminus S'$ is connected. Therefore, for every pair $(i,j)$ either one of $N_{i,j}, N'_{i,j}$ is empty or both are singletons.

Suppose that for every pair $(i,j)$ one of $N_{i,j}, N'_{i,j}$ is empty. Then at least $3$ of the $6$ sets are empty, and two of them must be in the same component, say $A$. Suppose that, for instance $N_{1,2}=N_{1,3}=\emptyset$. Then $N_A(s_1)=N_{1,2} \cup N_{1,3}=\emptyset$, a contradiction. Therefore, for at least one pair $(i,j)$, both $N_{i,j}$ and $N'_{i,j}$ are singletons. We can renumber the vertices of $S$ such that $N_{1,2}=\set{a}$ and $ N'_{1,2}=\set{a'}$ are singletons. Now suppose that for some other pair, say $(2,3)$, $N_{2,3}=\set{w_{23}}$ and $N'_{2,3}=\set{w'_{23}}$ are singletons. Then the matching $\set{a' w'_{23},s_1 a,s_3 w_{23}}$ disconnects $G$ into three odd components, contradicting Lemma \ref{lem:isolating} \itemref{OneOddEqimatchableComponent}. Therefore, both pairs $(2,3)$ and $(1,3)$ fall into the other category, i.e.\ one of $N_{2,3},N'_{2,3}$ and one of $N_{1,3},N'_{1,3}$ is empty. Since two sets from the same component cannot be empty, we conclude that without loss of generality $N_{2,3}=N'_{1,3}=\emptyset$. In other words, $A=N_{1,3}+a$ and $A'=N'_{2,3}+a'$. Hence, property \itemref{G32} also holds.
\end{proof}

We conclude this section with the following summarizing lemma.
\begin{lemma}\label{lem:kappa3}
Let $G$ be an equimatchable claw-free odd graph with $\alpha(G) \geq 3$. Then the following conditions are equivalent:
\begin{enumerate}[i)]
\item {$\kappa(G)=3$,}\label{itm:lemkappa3one}
\item {every independent set $S$ of $G$ with three vertices is a minimal cut set,} \label{itm:lemkappa3two}
\item {$G \in \cg_3$.}\label{itm:lemkappa3three}
\end{enumerate}
\end{lemma}
\begin{proof}
\itemref{lemkappa3one} $\Rightarrow$ \itemref{lemkappa3two} Let $S$ be an independent set with three vertices. By Lemma \ref{lem:notequiIS}, $S$ is a cut set, and since $\kappa(G)=3$ it is a minimal cut set.

\itemref{lemkappa3two} $\Rightarrow$ \itemref{lemkappa3three} By Proposition \ref{prop:3conn}.

\itemref{lemkappa3three} $\Rightarrow$ \itemref{lemkappa3one} We observe that $G_3$ contains only one 2-cut set, namely $\set{v_2, v_6}$. Since the multiplicities of $v_2$ and $v_6$ are at least two, this set does not yield a $2$-cut of $G$.
\end{proof}

\subsection{Equimatchable Claw-Free Odd Graphs with  $\alpha(G)\geq 3$ and $\kappa(G)=2$}\label{sec:2conn}
Throughout this section, $G$ is an equimatchable claw-free odd graph with $\alpha(G) \geq 3$ and $\kappa(G)=2$, $I$ is an independent set with three vertices, and $S=\set{s_1,s_2}$ is a (minimal) cut set of $G$. Recall that, by Corollary \ref{coro:ClawFreeTwoConnectedIsFactorCritical}, $G$ is factor-critical, and note that since $G$ is connected and $\alpha(G) \geq 3$, we have $n \ge 4$.
Our starting point is the following result on 2-connected equimatchable factor-critical graphs.

\begin{lemma}\label{lem:favaron}
\cite{favaron1986equimatchable} Let $G$ be a $2$-connected, equimatchable factor-critical graph with at least 4 vertices and $S=\set{s_1,s_2}$ be a minimal cut set of $G$. Then $G \setminus S$ has precisely two components, one of them even and the other odd. Let $A_S$ and $B_S$ denote the even and odd components of $G \setminus S$, respectively. Let $a_1$ and $a_2$ be two distinct vertices of $A_S$ adjacent to $s_1$ and $s_2$, respectively, and, if $\abs{B_S}>1$, let $b_1$ and $b_2$ be two distinct vertices of $B_S$ adjacent to $s_1$ and $s_2$, respectively. Then the following hold:
\begin{enumerate}
\item The subgraph $B_S$ is one of the four graphs $K_{2p+1}$, $K_{2p+1}- b_1b_2$, $K_{p,p+1}$, $K_{p,p+1} + b_1b_2$ for some $p \geq 1$. In the last two cases, all neighbors of $S$ in $B_S$ belong to the larger part of the bipartition of $K_{p,p+1}$.
\item The subgraph $A_S \setminus \set{a_1, a_2}$ is connected and randomly matchable, and if $\abs{B_S}>1$, then $A_S$ is connected and randomly matchable.
\end{enumerate}
\end{lemma}
In the rest of this section, $A_S, B_S$ denote the even and odd connected components of $G \setminus S$, respectively, and $a_1, a_2 \in A_S$ and $b_1, b_2 \in B_S$ are as described in Lemma \ref{lem:favaron}.
We note that the vertices $a_1, a_2$ exists, since otherwise $A_S$ contains a cut vertex of $G$.
Similarly, if $|B_S| \geq 3$, the vertices $b_1$ and $b_2$ exist.
Moreover, let $A'_S=A_S \setminus \set{a_1,a_2}$, and $B'_S=B_S \setminus \set{b_1,b_2}$ whenever $\abs{B_S} > 1$ (see Figure \ref{fig:favaron}). 
The minimal cut set $S$ is \emph{independent} if $s_1 s_2 \notin E(G)$, and \emph{strongly independent} if there exists an independent set $I$ with three vertices including $S$.

\begin{figure}[t]
\begin{center}
\includegraphics{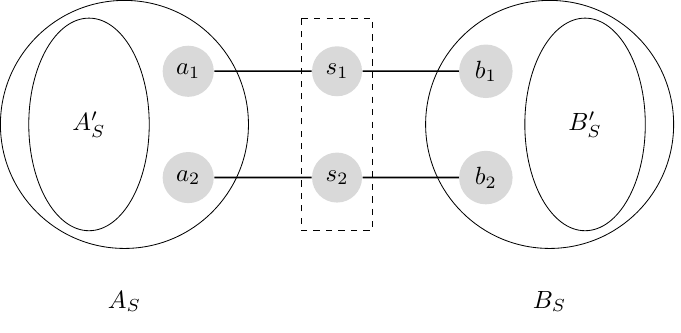}
\caption{The structure of $2$-connected equimatchable claw-free odd graphs by Lemma \ref{lem:favaron}.}\label{fig:favaron}
\end{center}
\end{figure}

An important consequence of Section \ref{sec:3conn} which will guide our proofs is the following:

\begin{corollary}\label{cor:stronglyind}(of Lemma \ref{lem:kappa3})
Let $G$ be an equimatchable claw-free odd graph with $\alpha(G)\geq 3$. If $\kappa(G)=2$, then it has a strongly independent 2-cut.
\end{corollary}
\begin{proof}
Since $\kappa(G)=2$, by Lemma \ref{lem:kappa3}, there exists an independent 3-cut $I$ that is not a minimal cut set, i.e.\ $I$ contains a minimal 2-cut $S \subseteq I$. Moreover, since $S \subseteq I$, $S$ is strongly independent.
\end{proof}

The main result of this section is that $G$ is either a $C_7$ or in one of the following graph families. The reader is referred to Figures \ref{subfig:G21}, \ref{subfig:G22} and \ref{subfig:G23} for these definitions. 

\begin{definition}
A graph is in $\cg_{21}$ if its vertex set can be partitioned into $V_1$ and $V_2$ such that
\begin{enumerate}[i)]
\item \label{itm:G211} $V_1$ induces a $K_{2q+1}$ for some $q \geq 1$,
\item \label{itm:G212} $V_2$ induces a $C_4$, say $v_1 v_2 v_3 v_4$,
\item \label{itm:G213} $N_{V_1}(v_1)=N_{V_1}(v_2)$,
\item \label{itm:G214} $2 \le \abs{N_{V_1}(v_1)} < \abs{V_1}$, and
\item \label{itm:G215} $N_{V_1}(v_3)=N_{V_1}(v_4)=\emptyset$.
\end{enumerate}

A graph is in $\cg_{22}$ if it has an independent 2-cut $S=\set{s_1,s_2}$ such that
\begin{enumerate}[i)]
\item \label{itm:G221} $A_S$ is a $K_{2p}$ for some $p \geq 1$,
\item \label{itm:G223} $B_S$ is a $K_{2q+1}$ for some $q \geq 0$,
\item \label{itm:G222} $s_1$ and $s_2$ are complete to $B_S$,
\item \label{itm:G224} $N_{A_S}(s_1) \cup N_{A_S}(s_2) \subsetneq A_S$, and
\item \label{itm:G225} $N_{A_S}(s_1) \cap N_{A_S}(s_2) = \emptyset$.
\end{enumerate}

A graph is in $\cg_{23}$ if it has an independent 2-cut $S=\set{s_1,s_2}$ such that
\begin{enumerate}[i)]
\item \label{itm:G231} $A_S$ is a $K_2$,
\item \label{itm:G232} $G[S \cup A_S]$ is a $P_4$,
\item \label{itm:G233} $B_S$ is a $K_{2q+1}$ for some $q \geq 1$, and
\item \label{itm:G234} $N_{B_S}(s_1) \cup N_{B_S}(s_2) = B_S$, $N_{B_S}(s_1)\neq \emptyset, N_{B_S}(s_2) \neq \emptyset$, either $N_{B_S}(s_1) \neq B_S$ or $N_{B_S}(s_2) \neq B_S$.
\end{enumerate}
\end{definition}

We note that
\begin{eqnarray*}
\cg_{21} & = & \set{G_{21}(1,1,1,1, x, 2q+1-x)|~ 2\leq x \leq 2q}, \\
\cg_{22} & = & \set{G_{22}(2p-x-y, x, y, 1,1, 2q+1)|~ q\geq 0, x,y \geq 1, x+y \leq 2p-1}, \\
\cg_{23} & = & \set{G_{23}(1,1,1,1, x, y, 2q+1-x-y)|~ 1 \leq x+y\leq 2q+1}
\end{eqnarray*}
where $G_{21}, G_{22}, G_{23}$ are the graphs depicted in Figures \ref{subfig:G21}, \ref{subfig:G22} and \ref{subfig:G23}, respectively. 
It can be noticed that the vertices $s_1$ and $s_2$ are not identified in $G_{21}$ of Figure \ref{subfig:G21} since the vertices playing the roles of $s_1$ and $s_2$ will depend on the case under analysis for this family.

\begin{proposition}\label{prop:2ConnectedOnlyIf}
If $G \in \cg_{21} \cup \cg_{22} \cup \cg_{23} + C_7$, then $G$ is a
connected equimatchable claw-free odd graph with $\alpha(G) \geq 3$ and $\kappa(G)=2$.
\end{proposition}
\begin{proof}
All the other properties being easily verifiable, we will check the equimatchability of a graph $G \in \cg_{21} \cup \cg_{22} \cup \cg_{23} + C_7$ by using Lemma \ref{lem:notequiIS}. One can observe that in each case, there is only one possible type (up to isomorphisms) of independent set $I$ of three vertices which is as described below.

If $G \in \cg_{21}$ then $I$ consists of $v_1,v_3$ and a vertex in $V_1$. Then $G \setminus I$ consists of one component with the single vertex $v_4$ and the other $G \setminus (I + v_4)$ which is odd.

If $G \in \cg_{22}$ then $I$ consists of $s_1,s_2$ and a vertex $a \in A_S \setminus (N_{A_S}(s_1)\cup N_{A_S}(s_2))$. Then $G \setminus I$ consists of two odd components, namely $B_S$ and $A_S - a$.

If $G \in \cg_{23}$ then $I = \set{a_1, s_2, b}$ where $b \in B_S \setminus N_{B_S}(s_2)$ (assuming without loss of generality that $x>0$). Then $G\setminus I$ consists of two odd components: the singleton $\set{a_2}$ and $G\setminus (I + a_2)$ which is odd.

Finally, if $G$ is a $C_7$, then for any independent set $I$ of three vertices, the graph $G \setminus I$ consists of two singletons and two adjacent vertices.
\end{proof}

In the rest of this section, we proceed as follows to prove the other direction: In Proposition \ref{prop:2ConnectedC4}, we analyze the case where $A_S$ is a $C_4$ for some 2-cut $S$. Subsequently, in Observation \ref{obs:favaronClawFree} we summarize Lemma \ref{lem:favaron} for the case where $A_S$ is not a $C_4$, and $\abs{B_S}>1$ where $S$ is an independent 2-cut. We further separate this case into two. In Proposition \ref{prop:2ConnectedGeneral}, we give the exact structure of $G$ when $B_S$ is neither a singleton nor a $P_3$. In Proposition \ref{prop:2ConnectedP3}, we give the exact structure of $G$ when $B_S$ is a $P_3$. We complete the analysis in Proposition \ref{prop:2ConnectedB1}, which determines the exact structure of $G$ in the last case, i.e.\ when $A_S$ is not a $C_4$ and $\abs{B_S}=1$. In the proofs of Propositions \ref{prop:2ConnectedGeneral}, \ref{prop:2ConnectedP3} and  \ref{prop:2ConnectedB1}, we heavily use the fact that the graph under consideration has a strongly independent 2-cut $S$. Moreover, this fact will allows us to conclude in Theorem \ref{thm:main} that we cover all possible cases for claw-free equimatchable odd graphs of connectivity 2.

\begin{proposition} \label{prop:2ConnectedC4}
If $A_S$ is a $C_4$ for some 2-cut $S$ of $G$, then $G \in \cg_{21}$ and $S$ is not independent.
\end{proposition}
\begin{proof}
Let $S=\set{s_1, s_2}$ be a 2-cut of $G$, and $A_S$ be a 4-cycle. 
In what follows, we show that $G \in \cg_{21}$ by setting $V_2=A_S$ and $V_1= V(G)\setminus V_2 = S \cup B_S$. 
Since $A_S$ is a 4-cycle, Property \itemref{G212} of $\cg_{21}$ holds for $G$. 
Let $A_S$ be the 4-cycle $v_1v_2v_3v_4$. 
By Lemma \ref{lem:ClawFreeCutSets} \itemref{TwoAdjacentNeigborsInCycle}, both $N_{A_S}(s_1)$ and $N_{A_S}(s_2)$ consist of two adjacent vertices of $V_2$. 
If $N_{V_2}(s_1) \neq N_{V_2}(s_2)$, then $N_{V_2}(s_1) \cup N_{V_2}(s_2)$ contains two non-adjacent vertices $x,y$ such that $x \in N_{V_2}(s_1)$, and $y \in N_{V_2}(s_2)$. 
Then the matching $\set{s_1x,s_2y}$ isolates the two vertices of $V_2 \setminus \set{x,y}$, contradicting Lemma \ref{lem:isolating} \itemref{OneOddEqimatchableComponent}.
Therefore, $N_{V_2}(s_1)=N_{V_2}(s_2)$ and it consists of two adjacent vertices of ${V_2}$, say $v_1$ and  $v_2$.
Since $S$ is a cut set, the neighbors of $v_1$ and $v_2$ in $V_1$ are exactly $s_1$ and $s_2$, thus showing \itemref{G213} and the first inequality of \itemref{G214}. 
The second part of the inequality follows from the fact that $S$ is a cut-set and $V_1 \setminus S = B_S \neq \emptyset$.
Furthermore, \itemref{G215} holds since $S\subsetneq V_1$ is a cut set and the neighborhood of $S$ in $V_2$ consists of $\{v_1, v_2\}$.

It remains to show Property \itemref{G211}, i.e.\ that $V_1$ is an odd clique.
Observe that $s_1s_2 \in E(G)$ since otherwise $S + v_2 + v_3$ forms a claw. 
Thus, $S$ is not independent.
The matching $\set{v_1v_4, v_2s_2}$ leaves the singleton $v_3$ as an odd component. 
Therefore, by Lemma \ref{lem:isolating} \itemref{OneOddEqimatchableComponent} and \itemref{EvenComponentsRandomlyMatchable}, $G[V_1-s_2]$ is randomly matchable which is either an even clique or a $C_4$ by Proposition \ref{prop:EvenEquimatchable}. 
Suppose that $G[V_1-s_2]$ is the cycle $s_1 b_1 b b_3$. 
We have that $N_{B_S}(s_1)=\set{b_1, b_3}$ is not a clique, contradicting Lemma \ref{lem:ClawFreeCutSets} \itemref{NeighborhoodOfConnectingVertexIsTwoCliques}. 
Therefore,  $G[V_1-s_2]$ is a $K_{2q}$ for some $q \geq 1$. By symmetry,  $G[V_1-s_1]$ is also a $K_{2q}$. Since $s_1s_2 \in E(G)$, we conclude that $G[V_1]$ is a $K_{2q+1}$ for some $q \geq 1$.
\end{proof}

\begin{observation}\label{obs:favaronClawFree}
If $S$ is an independent $2$-cut of $G$ and $\abs{B_S}>1$, then
\begin{enumerate}[i)]
\item The subgraph $G[A_S]$ is a $K_{2p}$ for some $p \geq 1$, and
\item The subgraph $G[B_S]$ is either a $K_{2q+1}$, or $K_{2q+1} - b_1b_2$ for some $q \geq 1$.
\end{enumerate}
\end{observation}
\begin{proof}
\begin{enumerate}[i)]
\item By Lemma \ref{lem:favaron}, $A_S$ is connected and randomly matchable. By Proposition \ref{prop:2ConnectedC4}, $A_S$ is not a $C_4$. Then, by Proposition \ref{prop:EvenEquimatchable}, $A_S$ is a $K_{2p}$ for some $p \geq 1$.
\item Recall Lemma \ref{lem:favaron}. In this case, $B_S$ cannot be a $K_{q,q+1}$ or $K_{q,q+1} + b_1b_2$ for $q \geq 2$ since otherwise (recalling that $b_1$ is in the larger part of the bipartition) $s_1$, $b_1$ and two vertices adjacent to $b_1$ in the smaller part of the bipartition of $B_S$ induce a claw. For $q=1$ we note that $K_{1,2}=K_3 - e$ and $K_{1,2}+e=K_3$. Therefore, $B_S$ is either a $K_{2q+1}$, or a $K_{2q+1} - b_1b_2$ for some $q \geq 1$.
\end{enumerate}
\end{proof}

\begin{figure}[p]
\begin{subfigure}{1\textwidth}
 \centering
  \includegraphics{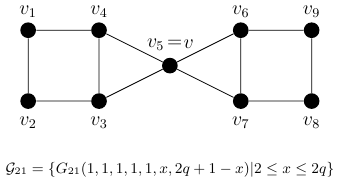}
  \caption{The graph $G_{11}$.}
  \label{subfig:G11}
\end{subfigure}
\begin{subfigure}{.5\textwidth}
  \centering
  \includegraphics{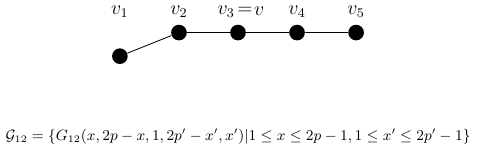}
  \caption{The graph $G_{12}$.}
  \label{subfig:G12}
\end{subfigure}
\begin{subfigure}{.5\textwidth}
  \centering
  \includegraphics{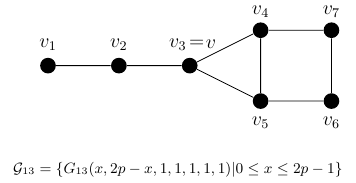}
  \caption{The graph $G_{13}$.}
  \label{subfig:G13}
\end{subfigure}
\begin{subfigure}{.5\textwidth}
  \centering
  \includegraphics{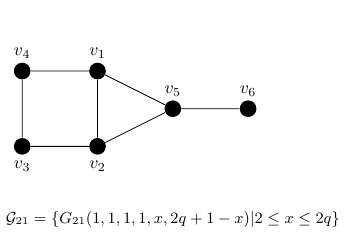}
  \caption{The graph $G_{21}$.}
  \label{subfig:G21}
\end{subfigure}
\begin{subfigure}{.5\textwidth}
  \centering
  \includegraphics{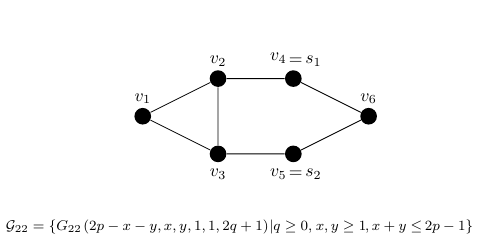}
  \caption{The graph $G_{22}$.}
  \label{subfig:G22}
\end{subfigure}
\begin{subfigure}{.5\textwidth}
  \centering
  \includegraphics{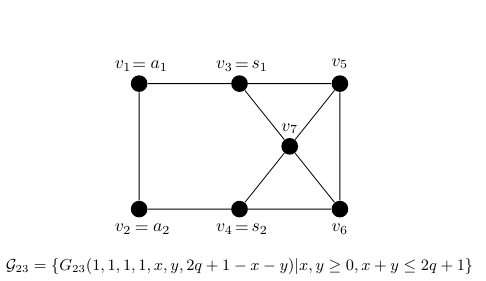}
  \caption{The graph $G_{23}$.}
  \label{subfig:G23}
\end{subfigure}
\begin{subfigure}{.5\textwidth}
  \centering
  \includegraphics{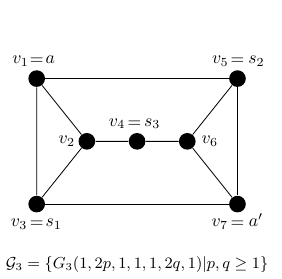}
  \caption{The graph $G_3$.}
  \label{subfig:G3}
\end{subfigure}
\caption{All equimatchable claw-free odd graphs with independence number at least 3 except $C_7$.}
\label{fig:graphfamilies}
\end{figure}

\begin{proposition}\label{prop:2ConnectedGeneral}
If there exists a strongly independent 2-cut $S$ of $G$ such that $\abs{B_S}>1$ and $B_S$ is not a $P_3$, then $G \in \cg_{22}$.
\end{proposition}
\begin{proof}
We now show that $G$ has all the properties of $\cg_{22}$. Let $S=\set{s_1, s_2}$ be a strongly independent 2-cut of $G$, and $I$ be an independent set of three vertices containing $S$.

The fact that $A_S$ is a $K_{2p}$ for some $p\geq 1$ ({Property \itemref{G221}}) follows from Observation \ref{obs:favaronClawFree}.

By the same observation and since $B_S$ is not a $P_3$, $B_S$ is either a $K_{2q+1}$ for some $q \geq 1$ or a $K_{2q+1}-b_1 b_2$ for some $q \geq 2$, thus 2-connected. 
Note that the unique vertex in $I \setminus S$ is not in $B_S$, since otherwise $G \setminus I$ consists of two even components. 
Therefore $I = S + a$ for some  $a \in A_S$. 
This implies that $N_{A_S}(s_1) \cup N_{A_S}(s_2) \subsetneq A_S$, thus {Property \itemref{G224}} is verified.

Now suppose that there exists a vertex $b \in B_S$ that is non-adjacent to $s_1$. 
Since both $A_S$ and $B_S$ are 2-connected, both of $A_S - a$ and $B_S - b$ are connected.
Moreover, $s_2$ is adjacent to $A_S - a$.
Then $I'= \set{a,s_1,b}$ is an independent set such that $G \setminus I'$ is either connected or consists of two even components $A_S-a+s_2$ and $B_S-b$, contradicting Lemma \ref{lem:notequiIS}. 
We conclude that $s_1$, and by symmetry $s_2$, are complete to $B_S$. This proves {Property \itemref{G222}}.

Since $s_1$ is complete to $B_S$, $B_S$ is a clique by Lemma \ref{lem:ClawFreeCutSets} \itemref{NeighborhoodOfConnectingVertexIsTwoCliques}. This shows {Property \itemref{G223}}.

Finally, {Property \itemref{G225}} follows from Property \itemref{G224} and Lemma \ref{lem:ClawFreeCutSets} \itemref{NeighborhoodOfCutSetInAClique}.
\end{proof}

\begin{proposition}\label{prop:2ConnectedP3}
If there exists some strongly independent 2-cut $S$ of $G$ such that $B_S$ is a $P_3$, then $G \in \cg_{23} + C_7$.
\end{proposition}

\begin{proof}
Let $S=\set{s_1, s_2}$ be a strongly independent 2-cut of $G$. We now show that $G$ is either a $C_7$ or has the following properties:
\begin{enumerate}[i)]
\item The subgraph $A_S$ is a $K_{2p}$ for some $p \geq 1$,
\item The subgraph $G[S \cup B_S]$ is a $P_5$,
\item $N_{A_S}(s_1) = N_{A_S}(s_2) = A_S$.
\end{enumerate}
Then, we will show that these properties imply that $G \in \cg_{23}$ with zero copies of $v_7$ and an independent 2-cut different from $S$.
\begin{enumerate}[i)]
\item Follows from Observation \ref{obs:favaronClawFree}.
\item The subgraph $G[B_S]$ is a path $b'_1 b' b'_2$. 
If $s_1$ is not adjacent to any of $b'_1$ and $b'_2$ then $s_1$ is adjacent to $b'$ and $B+s_1$ is a claw.
Therefore, without loss of generality $s_1$ is adjacent $b'_1$.
If $s_1 b'_2 \in E(G)$, then $N_B(s_1)$ is not a clique, contradicting Lemma \ref{lem:ClawFreeCutSets} \itemref{NeighborhoodOfConnectingVertexIsTwoCliques}. 
Therefore, $s_1 b'_2 \notin E(G)$.
Let $a$ be an arbitrary element of $A_S - a_2$.
Clearly, $A_S-a+s_2$ is connected.
If $s_1 b' \in E(G)$ then $I=\set{a,b'_1,b'_2}$ is an independent set such that $G \setminus I$ is either connected or has two even components. 
Therefore, $s_1 b' \notin E(G)$, concluding that $N_{B_S}(s_1) = \set{b'_1}$. 
Symmetrically, we have $N_{B_S}(s_2) = \set{b'_2}$.
Therefore, $b'_1=b_1$ and $b'_2=b_2$, thus $S \cup B_S$ induces the $P_5=s_1 b_1 b b_2 s_2$.

\item Recall that $A'_S= A_S - a_1 - a_2$.  
First assume that $A'_S \neq \emptyset$.
Furthermore, suppose that there is some $a' \in A'_S$ not adjacent to $s_1$. 
Then $I'=\set{s_1, a', b_2}$ is an independent set and $G \setminus I'$ has two even components, contradicting Lemma \ref{lem:notequiIS}. 
Therefore, $s_1$ is complete to $A'_S$ and symmetrically so is $s_2$. 
Now suppose that $s_1 a_2 \notin E(G)$, and consider the independent set $I''=\set{s_1, a_2, b_2}$. 
Then, $G \setminus I''$ has two even components, contradicting Lemma \ref{lem:notequiIS}. 
Therefore,  $s_1 a_2 \in E(G)$, and symmetrically $s_2 a_1 \in E(G)$. 
We conclude that $N_{A_S}(s_1)=N_{A_S}(s_2)=A_S$.

Now assume that $A'_S = \emptyset$, i.e. $A_S=\set{a_1,a_2}$. 
Then $a_1 a_2 s_2 b_2 b b_1 s_1$ is a Hamiltonian cycle of $G$. 
The edge set of $G$ possibly contains one or both of the edges $a_1 s_2, a_2 s_1$.
If both are edges of $G$, then $N_{A_S}(s_1)=N_{A_S}(s_2)=A_S$ and we are done. 
If none is an edge of $G$, then $G$ is a $C_7$.
We remain with the case that exactly one of $a_1 s_2, a_2 s_1$, say $a_1 s_2$ is an edge of $G$.
In this case $\set{a_2, s_1, b_2}$ is an independent set whose removal separates $G$ into two even components, contradicting Lemma \ref{lem:notequiIS}. 
\end{enumerate}
We now observe that the above properties imply $G \in \cg_{23}$. Indeed, let $S'$ be the independent set $\set{s_1, b_2}$, and verify the properties of $\cg_{23}$: \itemref{G231} $A_{S'}= \set{b,b_1}$ is a $K_2$, \itemref{G232} $G[S' \cup A_{S'}]=G[\set{s_1, b_2, b, b_1}]$ is the $P_4=s_1 b_1 b b_2$, \itemref{G233} $B_{S'}=A_{S}+s_2$ is an odd clique since $A_{S}$ is an even clique and $s_2$ is complete to it, \itemref{G234} $s_1$ is complete to $A_{S}$ and $b_2$ is adjacent to $s_2$, thus $N_{B_{S'}}(s_1)\cup N_{B_{S'}}(b_2)=B_{S'}$ and $N_{B_{S'}}(s_1), N_{B_{S'}}(b_2) \neq \emptyset$, furthermore $N_{B_{S'}}(s_1)\neq B_{S'}$ since $s_1s_2\notin E(G)$.
\end{proof}

\begin{proposition}\label{prop:2ConnectedB1}
If for every 2-cut $S$ of $G$ the component $A_S$ is not a $C_4$, and for every strongly independent 2-cut $S$ of $G$ the component $B_S$ consists of a single vertex, then $G \in \cg_{21} \cup \cg_{22} \cup \cg_{23}$.
\end{proposition}
\begin{proof}
Let $S=\set{s_1, s_2}$ be a strongly independent 2-cut of $G$. We remark that in this case we cannot use Observation \ref{obs:favaronClawFree}. Moreover, the only fact that we can deduce from Lemma \ref{lem:favaron} is that $A'_S$ is randomly matchable, a fact that is easily observed by applying Lemma \ref{lem:isolating} to the matching $\set{s_1 a_1, s_2 a_2}$.

We first observe that there are no 2-connected claw-free graphs on at most 5 vertices with an independent set of three vertices. Therefore, we can assume that $\abs{V(G)} > 5$, i.e.\ that $A'_S \neq \emptyset$.

We proceed with the proof by considering two disjoint cases.
\begin{itemize}
\item {$\mathbf{N_{A'_S}(s_1)=N_{A'_S}(s_2)=\emptyset:}$} In this case we will show that $G$ has all the properties of $\cg_{22}$. Properties \itemref{G223}, \itemref{G222}, \itemref{G224} clearly hold for $G$. If $A'_S$ is a $C_4$, then $S'=\set{a_1, a_2}$ is a 2-cut with $A_{S'}$ being a $C_4$, contradicting our assumptions. Therefore, $A'_S$ is a $K_{2p}$ for some $p \geq 1$. If $a_1 s_2 \in E(G)$, then $s_1, s_2, a_1$ and any neighbor of $a_1$ in $A'_S$ induce a claw, contradiction. Therefore, and using symmetry, we have that $a_1 s_2, a_2 s_1 \notin E(G)$, i.e.\ Property \itemref{G225} holds. It remains to show that $A_S$ is a clique.

    If $a_1 a_2 \notin E(G)$, then $S'=\set{a_1, a_2}$ is a strongly independent cut with $B_{S'}$ being a $P_3$, contradicting our assumptions. Therefore, $a_1 a_2 \in E(G)$. We now show that $A_S$ is a clique by proving that $a_1$ is complete to $A_{S}'$, and so is $a_2$ by symmetry. We first observe that $N_{A'_S}(a_1) \subseteq N_{A'_S}(a_2)$. Indeed, otherwise there is a vertex $a' \in A'_S$ adjacent to $a_1$ and not adjacent to $a_2$, and $\set{a_1,a_2,s_1,a'}$ induces a claw. By symmetry, we get $N_{A'_S}(a_1) = N_{A'_S}(a_2)$. This neighborhood has at least two vertices since otherwise $\kappa(G)=1$ where the unique common neighbor of $a_1$ and $a_2$ is a cut vertex. Now, suppose that $a_1$ is not complete to $A'_S$ and let $a' \in A'_S$ be non-adjacent to $a_1$. Then $I'=\set{a',a_1,s_2}$ is an independent set. Furthermore, $G \setminus I'$ consists of two even components, a contradiction to Lemma \ref{lem:notequiIS}. Therefore, $a_1$ is complete to $A'_S$, and so is $a_2$ by symmetry.

\item {$\mathbf{N_{A'_S}(s_1) \neq \emptyset:}$} 
We start by showing that $A_1=A'_S+a_1$ is a clique. 
Let $a'_1 \in N_{A'_S}(s_1)$  and apply Lemma \ref{lem:isolating} to the matching $\set{s_1 a'_1, s_2 a_2}$. 
It implies that $G[A'_S + a_1 - a'_1]$ is randomly matchable. 
Suppose that $G[A'_S + a_1 - a'_1]$ is a $C_4=a_1 a_2' a_3' a_4'$. 
Then $a_1 a_3', a_2' a_4' \notin E(G)$.
Then $A'_S$ is not a clique, thus it is the $C_4=a'_1 a'_2 a'_3 a'_4$.
By Lemma \ref{lem:ClawFreeCutSets} \itemref{TwoAdjacentNeigborsInCycle}, $N_{A'_S}(s_1)$ consists of two adjacent vertices of $A'_S$, namely $a'_1$ and without loss of generality $a'_2$.
Now, Lemma \ref{lem:isolating} applied to the matching $\set{s_1 a'_2, s_2 a_2}$ implies that $G[\set{a_1, a'_1, a'_3, a'_4}]$ is randomly matchable. 
However, $a'_1 a'_3 \notin E(G)$ and $a'_4$ is adjacent to all three vertices $a_1, a'_1$ and $a'_3$, thus, $G[\set{a_1, a'_1, a'_3, a'_4}]$ is neither a $C_4$ nor a clique, a contradiction.
Therefore, $G[A'_S + a_1 - a'_1]$ is a clique and consequently $A'_S$ is a $K_{2p}$ for some $p \geq 1$. 
This implies that $G[A'_S+a_1-a']$ is a $K_{2p}$ for every $a' \in N_{A'_S}(s_1)$, i.e.\ $a_1$ is complete to $A'_S-a'$. 
Moreover, $a_1$ is adjacent to $a'$ since $N_{A_S}(s_1)$ is a clique. 
We conclude that $a_1$ is complete to $A'_S$, i.e.\ that $A_1$ is a clique.

    Recall that $S$ is strongly independent. The unique vertex of $I \setminus S$ is some $a' \in A'_S \subseteq A_1$. By Lemma \ref{lem:ClawFreeCutSets} \itemref{NeighborhoodOfCutSetInAClique}, $N_{A_1}(s_1) \cap N_{A_1}(s_2)=\emptyset$. In particular, $a_1 s_2 \notin E(G)$. It remains to determine the neighborhoods of $a_2$ and $s_2$. We proceed by considering two disjoint cases regarding the neighborhood of $s_2$.
    \begin{itemize}
    \item {$\mathbf{N_{A'_S}(s_2) \neq \emptyset:}$} In this case, we will show that $G$ has all the properties of $\cg_{22}$. Properties \itemref{G223} and \itemref{G222} are trivial. Since the third vertex of $I$ is some $a' \in A'_S$, Property \itemref{G224} holds, too. It suffices to show that \itemref{G221} will hold, namely that $A_S$ is a clique. By Lemma \ref{lem:ClawFreeCutSets} \itemref{NeighborhoodOfCutSetInAClique}, this implies Property \itemref{G225}.

        Suppose that $a_2$ is not complete to $A_1$, and let $a$ be an arbitrary vertex of $A_1$ that is not adjacent to $a_2$. Then $I'=\set{a,a_2,b}$, where $b$ is the single vertex of the component $B_S$, is an independent set, and $a s_2 \notin E(G)$ since $N_{A_S}(s_2)$ is a clique by Lemma \ref{lem:ClawFreeCutSets} \itemref{NeighborhoodOfConnectingVertexIsTwoCliques}. Since $N_{A'_S}(s_2) \neq \emptyset$, $G \setminus I'$ is connected, a contradiction. Therefore, $a_2$ is complete to $A_1$, concluding that $A_S$ is a clique.

    \item {$\mathbf{N_{A'_S}(s_2) = \emptyset:}$} We first assume that $s_1 a_2 \notin E(G)$. In this case, we claim that for all $a'\in A'_S$ such that $s_1a'\notin E(G)$, $a_2$ is adjacent to $a'$. Indeed, if $a_2a'\notin E(G)$ for such a vertex $a'\in A'_S$ then $S'=\set{s_1,a_2}$ is a strongly independent 2-cut (contained by the independent set $\set{s_1,a_2,a'}$) with $|B_{S'}|\geq 3$, a contradiction to the assumption of this proposition. So, assume in what follows that $a_2$ is adjacent to every vertex in $A'_{S} \setminus N_{A'_{S}}(s_1)$. Now, we will show that $G$ has all the properties of $\cg_{23}$ using the independent 2-cut $S'=\set{s_1,a_2}$. Properties \itemref{G231}, \itemref{G232}, and \itemref{G233} are trivial since in this case $A_{S'}=\set{s_2, b}$ and $B_{S'} = A_1$. We now show Property \itemref{G234}. Since $a_2$ is adjacent to every vertex in $A'_{S} \setminus N_{A'_{S}}(s_1)$ and $s_1a_1\in E(G)$, we have that $N_{B_{S'}}(a_2) \cup N_{B_{S'}}(s_1)=B_{S'}$. Moreover, $N_{B_{S'}}(s_1) \neq \emptyset$ since $s_1a_1\in E(G)$. Finally, since  $\set{s_1,s_2}$ is a strongly independent 2-cut, there is a vertex $a'\in A'_S \subseteq A_1$ which is not adjacent to $s_1$ and consequently $a_2a'\in E(G)$ implying that $N_{B_{S'}}(a_2) \neq \emptyset$  and $N_{B_{S'}}(s_1) \neq B_{S'}$.

        Now assume that $s_1 a_2 \in E(G)$. In this case, we set $V_1=A_1$ and show that $G$ has all the properties of $\cg_{21}$. Property \itemref{G211} holds since $A_1$ is a clique, and \itemref{G212} holds since $V(G) \setminus A_1$ is the cycle $s_1 a_2 s_2 b$. Property \itemref{G215} holds since $b$ and $s_2$ do not have neighbors in $A_1$. We now show that \itemref{G213} holds. $N_{A_1}(a_2) \subseteq N_{A_1}(s_1)$ since otherwise $a_2, s_1, s_2$ and a fourth vertex that is adjacent to $a_2$ and non-adjacent to $s_1$ form a claw. Furthermore, $N_{A_1}(s_1) \subseteq N_{A_1}(a_2)$ since otherwise $s_1, a_2, b$ and a fourth vertex adjacent to $s_1$ and non-adjacent to $a_2$ form a claw. We now proceed to Property \itemref{G214}. Since $N_{A'_S}(s_1) \neq \emptyset$ and $s_1a_1\in E(G)$, we have $\abs{N_{A_1}(s_1)} \geq 2$. Moreover, $N_{A_1}(s_1) \neq A_1$ since otherwise $\alpha(G)=2$. This concludes the proof.
    \end{itemize}
\end{itemize}
\end{proof}

Let us summarize the results of this section in the following:

\begin{proposition}\label{prop:KappaTwo}
If $G$ is an equimatchable claw-free odd graph with $\alpha(G) \geq 3$ and $\kappa(G)=2$, then $G \in \cg_{21} \cup \cg_{22} \cup \cg_{23} + C_7$.
\end{proposition}
\begin{proof}
Let $S$ be a 2-cut of $G$. By Lemma \ref{lem:favaron}, $G \setminus S$ consists of an even component $A_S$ and an odd component $B_S$.  Proposition \ref{prop:2ConnectedC4} proves that if for some 2-cut $S$ we have that $A_S$ is a $C_4$, then $G \in \cg_{21}$. In what follows we assume that for every 2-cut $S$ of $G$, $A_S$ is not a $C_4$.

By Corollary \ref{cor:stronglyind}, $G$ contains a strongly independent 2-cut. We consider the set $\cs \neq \emptyset$ of all the strongly independent (minimal) 2-cuts, and consider the following disjoint and complementing subcases:
\begin{itemize}
\item{There exists some $S' \in \cs$ such that $\abs{B_{S'}} > 1$ and $B_{S'}$ is not a $P_3$.} In this case by Proposition \ref{prop:2ConnectedGeneral}, $G \in \cg_{22}$.
\item{There exists some $S' \in \cs$ such that $B_{S'}$ is a $P_3$.} In this case, by Proposition \ref{prop:2ConnectedP3}, $G$ is either a $C_7$ or a graph of $\cg_{23}$.
\item{$\abs{B_{S'}}=1$ for every $S' \in \cs$.} In this case, by Proposition \ref{prop:2ConnectedB1}, we have that $G \in \cg_{21} \cup \cg_{22} \cup \cg_{23}$.
\end{itemize}

\end{proof}

\subsection{Equimatchable Claw-Free Odd Graphs with  $\alpha(G)\geq 3$ and $\kappa(G)=1$}\label{sec:1conn}
Let us finally consider equimatchable claw-free odd graphs with independence number at least 3 and connectivity 1. We will show that these graphs fall into the following family.

\begin{definition} \label{defn:GraphFamilyOneConn}
Graph $G \in \cg_1$ if it has a cut vertex $v$ where $G - v$ consists of two connected components $G_1, G_2$ such that for $i \in \set{1,2}$

\begin{enumerate}[i)]
\item \label{itm:G11} Component $G_i$ is either an even clique or a $C_4$.
\item \label{itm:G12} If $G_i$ is a $C_4$, then $N_{G_i}(v)$ consists of two adjacent vertices of $G_i$.
\item \label{itm:G13} If both $G_1$ and $G_2$ are cliques, then $v$ has at least one non-neighbor in each one of $G_1$ and $G_2$.
\end{enumerate}
\end{definition}
We note that $\cg_1 = \set{G_{11}} \cup \cg_{12} \cup \cg_{13}$ where
\begin{eqnarray*}
\cg_{12} & = & \set{G_{12}(x,2p-x,1,2p'-x',x')|~ 1\leq x \leq 2p-1 , 1\leq x' \leq 2p'-1},\\
\cg_{13} & = & \set{G_{13}(x,2p-x,1,1,1,1,1)|~ 0\leq  x \leq 2p-1}
\end{eqnarray*}
where $G_{11}, G_{12}, G_{13}$ are the graphs depicted in Figures \ref{subfig:G11}, \ref{subfig:G12} and \ref{subfig:G13}, respectively.

\begin{proposition}\label{prop:1ConnectedOnlyIf}
If $G \in \cg_1$, then $G$ is a
connected equimatchable claw-free odd graph with $\alpha(G) \geq 3$ and $\kappa(G)=1$.
\end{proposition}
\begin{proof}
All the other properties being easily verifiable, we will only show that $G$ is equimatchable using Lemma \ref{lem:notequiIS}. Note that $V(G_i) \setminus N(v)$ is a non-empty clique where $v$ is a cut vertex of $G$. Therefore, every independent set $I$ with three vertices containing $v$ has exactly one vertex from every $G_i$. In this case, $G \setminus I$ has two odd components. An independent set $I'$ with three vertices that does not contain $v$ must contain two non-adjacent vertices of a $C_4$ and one vertex from the other component. Then one vertex of that $C_4$ is isolated in $G \setminus I'$. Let $v'$ be the unique vertex of $I'$ in the other component $G_i$. If $v$ is a cut vertex of $G$ (which happens when $G_{i}$ is an even clique and $N_{G_i}(v)=\set{v'}$), then $G_{i}-v'$ constitutes a second odd connected component of $G \setminus I'$; otherwise, $G \setminus I'$ consists of two connected components and they are both odd.
\end{proof}

\begin{proposition}\label{prop:KappaOne}
If $G$ is an equimatchable claw-free odd graph with $\alpha(G) \geq 3$ and $\kappa(G)=1$, then $G \in \cg_1$.
\end{proposition}
\begin{proof}
By Lemma \ref{lem:ClawFreeCutSets} \itemref{VertexAdjacentToAtmostTwoComponents}, every cut vertex of $G$ separates it into two connected components $G_1$ and $G_2$. From parity considerations, $G_1$ and $G_2$ are either both even or both odd. We consider two complementing cases:
\begin{itemize}
\item {\textbf{Graph $G$ has a cut vertex $v$ such that $G_1$ and $G_2$ are even.}} Let $u$ be a vertex of $G_1$ adjacent to $v$. Considering the matching $M$ consisting of the single edge $uv$ and applying Lemma \ref{lem:isolating} \itemref{EvenComponentsRandomlyMatchable}, we conclude that $G_2$ is randomly matchable, i.e., either an even clique or a $C_4$ by Proposition \ref{prop:EvenEquimatchable}. By symmetry, the same holds for $G_1$; thus, \itemref{G11} in Definition \ref{defn:GraphFamilyOneConn} holds. Assume that $G_i$ is a $C_4$ for some $i \in \set{1,2}$. Then, by Lemma \ref{lem:ClawFreeCutSets} \itemref{TwoAdjacentNeigborsInCycle}, $v$ is adjacent to exactly two adjacent vertices of $G_i$; thus, \itemref{G12} in Definition \ref{defn:GraphFamilyOneConn} holds. Finally, since $\alpha(G)\geq 3$, \itemref{G13} in Definition \ref{defn:GraphFamilyOneConn} also holds.

\item {\textbf{Every cut vertex $v$ of $G$ separates it into two odd components.}} We will conclude the proof by showing that this case is not possible. No two cut vertices of $G$ are adjacent, since otherwise one of them disconnects $G$ into two even components. Let $v$ be a cut vertex, $G_1$ and $G_2$ be the connected components of $G - v$, and $u_1$ be a neighbor of $v$ in $G_1$. Applying Lemma \ref{lem:isolating}  \itemref{EvenComponentsRandomlyMatchable} to the matching consisting of the single edge $u_1v$, we conclude that $G_1 - u_1$ is randomly matchable. Then, either $G_1 - u_1$ is connected, or by Lemma \ref{lem:ClawFreeCutSets} \itemref{VertexAdjacentToAtmostTwoComponents}, $G_1 - u_1$ has exactly two connected components. Moreover, since $u_1$ is not a cut vertex of $G$, $v$ has a neighbor in each of these components. If there are two such components, the neighbors of $v$ in these components do not form a clique, contradicting Lemma \ref{lem:ClawFreeCutSets} \itemref{NeighborhoodOfConnectingVertexIsTwoCliques}. Therefore, $G_1 - u_1$ is connected, and by Proposition \ref{prop:EvenEquimatchable} we conclude that it is either a $C_4$ or an even clique.

    Suppose that $G_1-u_1$ is a $C_4$, say $w_1 w_2 w_3 w_4$. By Lemma \ref{lem:ClawFreeCutSets} \itemref{TwoAdjacentNeigborsInCycle}, $N_{G_1-u_1}(v)$ consists of two adjacent vertices, say $w_1,w_2$. Consider the matching $M=\set{vw_1, w_2 w_3}$. $V(M)$ disconnects $\set{u_1,w_4}$ from $G$. If $u_1$ and $w_4$ are non-adjacent, they contradict Lemma \ref{lem:isolating} \itemref{OneOddEqimatchableComponent}. Therefore, $u_1$ is adjacent to $w_4$. Now the matching $M=\set{v w_2, u_1 w_4}$ disconnects the vertices $w_1$ and $w_3$ from $G$ and leaves two odd components, contradicting Lemma \ref{lem:isolating} \itemref{OneOddEqimatchableComponent}. Hence, we conclude that $G_1-u_1$ cannot be a $C_4$ and therefore has to be an even clique.
    
    We now show that $u_1$ is complete to $G_1 - u_1$. Suppose that there exists a vertex $z$ of $G_1 - u_1$ that is non-adjacent to $u_1$. Then $z$ is non-adjacent to $v$ since otherwise $v$ has two non-adjacent vertices, namely $u_1$ and $z$, in its neighborhood in $G_1$, a contradiction by Lemma \ref{lem:ClawFreeCutSets} \itemref{NeighborhoodOfConnectingVertexIsTwoCliques}. Let $z'$ be a vertex of $G_1 - u_1$ that is adjacent to $v$. Recall that such a vertex exists since $u_1$ is not a cut vertex of $G$, and clearly, $z \neq z'$. Now consider the matching consisting of the edge $vz'$ and a perfect matching of the even clique $G_1 \setminus \set{u_1,z,z'}$. This matching leaves $u_1$ and $z$ as two odd components, a contradiction by Lemma \ref{lem:isolating} \itemref{OneOddEqimatchableComponent}. Therefore, $G_1$ is an odd clique, and $v$ is adjacent to at least two vertices (namely, $u_1$ and $z'$) of $G_1$. By symmetry, the same holds for $G_2$.

Since $\alpha(G)\geq 3$, $v$ is not adjacent to some vertex $w_1$ of $G_1$ and some vertex $w_2$ of $G_2$. Then $S=\set{v,w_1,w_2}$ is an independent set of $G$ and $G \setminus S$ consists of two even components, contradicting Lemma \ref{lem:notequiIS}.
\end{itemize}
\end{proof}

\section{Summary and Recognition Algorithm}\label{sec:Summary}
In this section we summarize our results in Theorem \ref{thm:main} and use it to develop an efficient recognition algorithm.

\begin{theorem}\label{thm:main}
A graph $G$ is a connected claw-free equimatchable graph if and only if one of the following holds:
\begin{enumerate}[i)]
\item \label{itm:TheoremC4} $G$ is a $C_4$.

\item \label{itm:TheoremEvenClique} $G$ is a $K_{2p}$ for some $p \geq 1$.

\item \label{itm:TheoremAlpha2} $G$ is odd and $\alpha(G)\leq 2$.

\item \label{itm:TheoremKappa1} $G \in \cg_1$.

\item \label{itm:TheoremKappa2} $G \in \cg_{21} \cup \cg_{22} \cup \cg_{23} + C_7$.

\item \label{itm:TheoremKappa3} $G \in \cg_3$.

\end{enumerate}
\end{theorem}
\begin{proof}
One direction follows from Proposition \ref{prop:EvenEquimatchable}, Lemma \ref{lem:EibenAlphaTwo} and Propositions \ref{prop:1ConnectedOnlyIf}, \ref{prop:2ConnectedOnlyIf} and \ref{prop:3ConnectedOnlyIf} in the order of the items from $\itemref{TheoremC4}$ to $\itemref{TheoremKappa3}$. We proceed with the other direction. Let $G$ be an equimatchable claw-free graph. If $G$ is even, then by Proposition \ref{prop:EvenEquimatchable}, it is either a $C_4$ or an even clique. It remains to show that if $G$ is odd and $\alpha(G) \geq 3$, then $G$ is either a $C_7$ or in one of the families $\cg_1, \cg_{21}, \cg_{22}, \cg_{23}, \cg_3$. If $\kappa(G)=1$, then $G \in \cg_1$ by Proposition \ref{prop:KappaOne}. If $\kappa(G)=2$, then $G \in \cg_{21} \cup \cg_{22} \cup \cg_{23} + C_7$ by Proposition \ref{prop:KappaTwo}. If $\kappa(G)=3$, then $G \in \cg_3$ by Proposition \ref{prop:3conn}.
\end{proof}

The recognition problem of claw-free equimatchable graphs is clearly polynomial since each one of the properties can be tested in polynomial time. Equimatchable graphs can be recognized in time \bigoh{m \bar{m}} (see \cite{demange_ekim_equi}), where $m$ (respectively $\bar{m})$ is the number of edges (respectively non-edges) of the graph. Claw-free graphs can be recognized in \bigoh{m^{\frac{\omega+1}{2}}} time, where $\omega$ is the exponent of the matrix multiplication complexity (see \cite{Kloks2000}). The currently best exponent for matrix multiplication is $\omega \approx 2.37286$ (see \cite{LeGall}), yielding an overall complexity of \bigoh{ m (\bar{m} + m^{0.687})}.

We now show that our characterization yields a more efficient recognition algorithm.

\alglanguage{pseudocode}

\begin{algorithm}[H]
\caption{Claw-free equimatchable graph recognition}\label{alg:recognition}
\begin{algorithmic}[1]
\Require{A graph $G$.}
\Statex
\If {$G$ is even}
\State \Return ($G$ is a clique or $G$ is a $C_4$).\label{step:CliqueAndC4}
\EndIf

\If {$\bar{G}$ is triangle free}
\Return \textbf{true}.\label{step:TriangleFree}
\EndIf

\If {$G$ is a $C_7$ or $G$ is a $G_{11}$}
\Return \textbf{true}.
\EndIf

\Statex
\State Compute the unique twin-free graph $H$ and multiplicities $n_1, \ldots, n_k$ such that
$G=H(n_1, \ldots, n_k)$. \label{step:TwinElimination}

\If {$H$ is isomorphic to neither one of $G_{12}, G_{13}, G_{21}, G_{22}, G_{23}, G_3$ nor to a relevant subgraph of it}  \label{step:Isomorphism}
\State \Return \textbf{false}
\Else
\State let $H$ be isomorphic to $G_x$ for some $x \in \set{12, 13, 21, 22, 23, 3}$ or to a relevant subgraph of it.
\EndIf

\State \Return \textbf{true} if and only if $(n_1,\ldots,n_k)$ matches the multiplicity pattern in the definition of $\cg_x$. \label{step:CheckMultiplicities}
\end{algorithmic}
\end{algorithm}

\begin{corollary}\label{cor:recognition}
Algorithm \ref{alg:recognition} can recognize equimatchable claw-free graphs in time \bigoh{ m^{1.407}}.
\end{corollary}
\begin{proof}
The correctness of Algorithm \ref{alg:recognition} is a direct consequence of Theorem \ref{thm:main}. As for its time complexity, Step \ref{step:CliqueAndC4} can be clearly performed in linear time. Step \ref{step:TriangleFree} can be performed in time \bigoh{ m^{\frac{2 \omega}{\omega+1}}} = \bigoh{ m^{1.407}} (see \cite{Alon1997}).

For every graph $G$ there is a unique twin-free graph $H$ and a unique vector $(n_1, \ldots, n_k)$ of vertex multiplicities such that $G=H(n_1, \ldots, n_k)$.
The graph $H$ and the vector $(n_1,\ldots,n_k)$ can be computed from $G$ in linear time using partition refinement, i.e.\ starting from the trivial partition consisting of one set, and iteratively refining this partition using the closed neighborhoods of the vertices (see \cite{HabibPV99}).
Each set of the resulting partition constitutes a set of twins.
Therefore, Step \ref{step:TwinElimination} can be performed in linear time.

We now note that for some values of $x \in \set{12, 13, 21, 22, 23, 3}$, at most one entry of the multiplicity vector is allowed to be zero. In this case $H$ is not isomorphic to $G_x$ but to an induced subgraph of it with one specific vertex removed. We refer to these graphs as \emph{relevant subgraphs} in the algorithm.

As for Step \ref{step:Isomorphism}, it takes a constant time to decide whether an isomorphism exists: if $H$ has more than 9 vertices, it is isomorphic to neither one of $G_{11}, G_{12}, G_{13}, G_{21}, G_{22}, G_{23}, G_3$ nor to a subgraph of them; otherwise, $H$ has to be compared to each one of these graphs and their relevant subgraphs, where each comparison takes constant time. Finally, Step \ref{step:CheckMultiplicities} can be performed in constant time.

We conclude that the running time of Algorithm \ref{alg:recognition} is dominated by the running time of Step \ref{step:TriangleFree}, i.e.\ \bigoh{ m^{1.407}}.
\end{proof}


\end{document}